\documentclass[11pt]{article}
\usepackage[utf8]{inputenc}
\usepackage[english]{babel}
\usepackage[active]{srcltx}
\usepackage{hyperref}
\usepackage[all]{xy}
\usepackage{amsfonts}
\usepackage{amsthm}
\usepackage{enumerate}
\usepackage{mathrsfs}
\usepackage{multicol}
\usepackage[all]{xy}
\usepackage{amsfonts}
\usepackage{latexsym,amsmath,amssymb}
\usepackage[usenames,dvipsnames]{color}
\usepackage{pstricks}
\usepackage[absolute,overlay]{textpos}
\usepackage{graphics,wrapfig,times}
\usepackage{xfrac}
\usepackage{appendix}
\usepackage{yhmath}
\usepackage{amsmath}

\usepackage{color}
\definecolor{DarkBlue}{rgb}{0.1,0.1,0.5}
\definecolor{Red}{rgb}{0.9,0.1,0.1}
\definecolor{Green}{rgb}{0.3,0.7,0.0}
\definecolor{green2}{rgb}{0.1,0.7,0.2}
\definecolor{blue2}{rgb}{0.0,0.6,0.7}
\definecolor{pink}{rgb}{1,0.0,1}
\definecolor{orange}{rgb}{0.9,0.0,0.1}

\newtheorem{theorem}{Theorem}
\newtheorem{corollary}{Corollary}

\newtheorem{lemma}{Lemma}
\newtheorem{proposition}{Proposition}

\newtheorem{definition}{Definition}
\newtheorem{remark}{Remark}
\newtheorem{example}{Example}

\renewcommand{\d}{\mathrm{d}}

\renewcommand{\d}{\mathrm{d}}

\newcommand{\derpar}[2]{\displaystyle\frac{\partial{#1}}{\partial{#2}}}

\newcommand{\df}{\Omega}

\newcommand{\Tan}{\mathrm{T}}
\newcommand{\inn}{{\mathop{i}\nolimits}}
\newcommand{\Lie}{\mathop{\mathrm{L}}\nolimits}

\newcommand{\bal}{\begin{align*}}
\newcommand{\eal}{\end{align*}}
\def\beq{\begin{equation}}
\def\eeq{\end{equation}}
\def\bea{\begin{eqnarray}}
\def\eea{\end{eqnarray}}
\def\beann{\begin{eqnarray*}}
\def\eeann{\end{eqnarray*}}
\def\ben{\begin{enumerate}}
\def\een{\end{enumerate}}
\def\bit{\begin{itemize}}
\def\eit{\end{itemize}}

\newcommand{\bi}{\begin{itemize} } 
\newcommand{\ei}{\end{itemize} } 
\newcommand{\be}{\begin{enumerate} } 
\newcommand{\ee}{\end{enumerate} } 

\def\df{{\mit\Omega}}

\def\d{{\rm d}}

\def\Tan{{\rm T}}
\def\Lie{\mathop{\rm L}\nolimits}
\def\inn{\mathop{i}\nolimits}
\def\Cinfty{{\rm C}^\infty}
\DeclareMathOperator{\Ima}{Im}

\title{Geometric Gauge Freedom in Multisymplectic Field Theories}
\author{\sffamily 
$^a$Jordi Gaset
\thanks{jordi.gaset@unir.net\quad ORCID: 0000-0001-8796-3149}
\\[1ex]
\normalsize\itshape\sffamily 
$^a$Escuela Superior de Ingeniería y Tecnología,
Universidad Internacional de La Rioja,
Logroño, Spain.
}

\begin{document}

\maketitle

\pagestyle{myheadings}

\thispagestyle{empty}

\begin{abstract}
We use the kernel of a premultisymplecic form to classify its solutions, inspired by the work of M. Gotay and J. Nester. In the case of variational premultisymplectic forms, there is an equivalence relation which classify the solutions in general distributions called expanded solutions. We also present an equivalence relation for sections and a reduction procedure for the system. We apply this results to mechanics, electromagnetism and Metric-Affine gravity, and compare them with other notions of gauge freedom.
\end{abstract}


\tableofcontents

\section{Introduction}

Multisymplectic manifolds are one of the most successfully geometric frameworks for classical field theories. The reduction of multisymplectic systems by symmetries reminds one of the most active and relevant fields of research in multisymplectic geometry \cite{ blacker_reduction_2020, blacker_reduction_2022,echeverr_remarks_2018, lopez_reduction_1999,madsen_multi-moment_2012,marsden_covariant_1986,sniatycki_multisymplectic_2004}.

Reducing field theories is inherently more difficult than reducing mechanical theories. One of the main reasons is that, in a regular mechanical system, there is one solution for every initial conditions, but a regular multisymplectic field theory usually have more. In other words, a regular mechanical systems has one vector field as a solution, meanwhile a multisymplectic field theory has several distributions as a solution. Consequently, the generalization of the methods in mechanics has a limited application in field theories.

If we want to study and classify the solutions of a mechanical system from a geometric point of view, there are three main steps. First, we should proof that there exists a solution. If not, one can perform a constraint algorithm to find a submanifold were tangent solutions exist. 

Next, one should check if the solutions (understood as a vector field), is unique. This situations was studied by M. Gotay and J. Nester \cite{gotay_presymplectic_1979}, where they show that the multiplicity of solutions is characterized by the vector fields of the kernel of the presymplectic form. They called these elements gauge vector fields (as they describe several solutions which can be interpreted as redundant information). The also present a gauge vector field algorithm to find ''hidden" vector fields by computing Lie parenthesis with a solution. If the distribution for gauge vector fields has good properties, one can reduce the system by quoting the manifold by the gauge distribution, obtaining a regular symplectic form.

Finally, mechanical system can have symmetries, which we can use to simplify the problem. Usually, on considers the restricted system in a submanifold where a conserved quantity is constant. A geometric implementation of this idea is the Marsden-Weinstein reduction \cite{marsden_reduction_1974}. This theorem provides sufficient conditions to have a symplectic system in the submanifold, such that the solutions from the original and the reduced systems are compatible. In the case that the symplectic form becomes presymplectic on the submanifold, the algorithm prescribes a quotient to eliminate the degeneracy of the presymplectic form, like in the Gotay-Nester algorithm.

We have a similar situation in multisymplectic field theories. If the system has no solutions, there is a constraint algorithm which provides a submanifold where tangent solutions exists \cite{de_leon_pre-multisymplectic_2005}. There has been a lot of work to generalize the Marsden-Weinstein algorithm to the multisymplectic setting (see, for instance, \cite{blacker_reduction_2020,echeverr_remarks_2018, sniatycki_multisymplectic_2004}). 

There has been less work in the generalization of Gotay-Nester algorithm to multisymplectic systems. There are two good reasons for it. First, there can be multiple solutions even for regular systems, which cannot happen in mechanics. In other words, there is a multiplicity of solutions which is not described by the kernel of the (pre)multisymplectic form. The other problem is that the gauge vector field algorithm uses the Lie bracket of gauge vector fields and the vector field solution. In multisymplectic systems the solutions are given by distributions (or multivector fields), so this step of the algorithm cannot be transported to multisymplectic geometry. Nevertheless, one can consider the quotient by the kernel of the premultisymplectic form in order to build a regular multisymplectic system \cite{echeverr_remarks_2018}.

In this work we will present a classification of the solutions of a premultisymplectic systems characterized by the kernel of the premultisymplectic form.

\section{Multisymplectic formalism}


This is a brief introduction to the multisymplectic formalism and transverse distributions. For more details on the multisymplectic formalism see \cite{gaset_variational_2016}.

Let  $\pi\colon{J}\to M$ be a fibre bundle, where
$\dim\,M=m\geq 1$ and $\dim\,{J}=n+m$, and $M$ is an
connected orientable manifold with volume form $\eta\in\df^m(M)$. We write $(U;x^\mu,y^j)$,
$\mu=1,\ldots ,m$, $j=1,\ldots ,n$, for local charts of
coordinates in $J$ adapted to the fibred structure, and such that
$\pi^*\eta=\d x^1\wedge\ldots\wedge\d x^m\equiv\d^mx$.

We denote by $V(\pi)$ the distribution of vertical vectors, that is, $V(\pi):=\cup_{p\in J}Ker (T_p\pi)$. The vertical vector fields, which are the sections of the vertical distribution, are denoted $\Gamma (V(\pi))$. They are locally generated by $ \left\{\sfrac{\partial}{\partial y^j}\right\}$. 

\begin{definition}
A generalized distribution $D\subset TJ$ is {\rm transverse} if it has a subdistribution $H\subset D$ which is a smooth distribution of constant rank $m$ such that $T\pi(H)=TM$.
\end{definition}

A transverse distribution has rank greater or equal to $m$ at every point. As the restriction of $T\pi$ to $D$, $T\pi|_D:D\rightarrow TM$, is exhaustive if $D$ is transverse, we also have the decomposition:

\begin{lemma}\label{lem:descomposa}
A transverse distribution $D$ can be decomposed as
$$D=H+ B,$$
where $H$ is a transverse distribution of rank $m$ and $B$ is a generalized distribution such that $B\subset V(\pi)$.
\end{lemma}

 \begin{definition}
A form $\Omega\in\df^{m+1}({J})$ ($m\geq 1$) is a
{\rm multisymplectic form} if it is closed and $1$-nondegenerate, that is, if the map 
$\flat_{\Omega}\colon\Tan {J} \longrightarrow \Lambda^m\Tan^*{J}$,
defined by $\flat_{\Omega}(x,v)=(x,\inn(v)\Omega_x)$,
for every $x\in{J}$ and $v\in\Tan_x{J}$, is injective. 
If the form is only closed, then it is a {\rm premultisymplectic form}.
 \end{definition}

Some of the results that appear in this work will only hold an special kind of multisymplectic forms. 
 
 \begin{definition} A form $\omega$ is \emph{variational} if
$$
\inn(Z_1)\inn(Z_2)\inn(Z_3)\Omega=0\ ,\text{ for every }
Z_1,Z_2,Z_3\in \Gamma(V(\pi))\,.
$$
\end{definition}

The name "variational" is justified because this is the situation in
the Lagrangian and Hamiltonian formalism of field theories, including the restricted forms obtained after a quotient or considering constraints. It is relevant for multisymplectic \cite{gaset_variational_2016} and also multicontact structures \cite{LGMRR_2022}. In a chart of adapted coordinates defined on an open set $U$, a variational multisymplectic form can be expressed as
\beq
\Omega\vert_U=\d F_j^\mu\wedge\d y^j\wedge\d^{m-1}x_\mu+\d E\wedge\d^m x \ ,
\label{omega}
\eeq
where $\displaystyle\d^{m-1}x_\mu=\inn\left(\derpar{}{x^\mu}\right)\d^mx$,
and
$F_j^\mu(x^\nu,y^i),E(x^\nu,y^i)\in\Cinfty(U)$.

At a point $p\in J$ we can define a family of kernels of $\Omega$, indexed by $1\leq k\leq m+1$:
$$
\ker^k_p \Omega=\{\mathbf{X}\in\Lambda^k T_p^*J | \inn(\mathbf{X})\Omega|_p=0 \}\,.
$$
As usual, we also define $\ker^k\Omega:=\cup_{p\in J}\ker^k_p \Omega$. In this work we will use the cases $k=1$ and $k=m$.

A (pre)multsiymplectic systems $(J,\Omega)$ is a fiber bundle endowed with a (pre)multisympelctic form $\Omega$.

\begin{definition}\label{dfn:sol} Let $(J,\Omega)$ be a (pre)multisymplectic system.
\begin{itemize}
    \item A distribution $D$ of rank $m$ is a \emph{solution} of the system if it is transverse, integrable, and $\Lambda^mD\subset \ker^m\Omega$
    \item A section $\phi:M\rightarrow J$ is a \emph{solution} of the system if
    $$
    \phi^*\inn(Y)\Omega=0\,,\quad\text{ for every } Y\in \mathfrak{X}(J)\,.
    $$
\end{itemize}
\end{definition}

If $D$ is a solution, then its integrable sections are also solutions.  It will be more convenient  to ignore the integrability condition for distributions, thus, we denote by $S(\Omega)$ the set of smooth distributions of rank $m$ which are transverse and $\Lambda^mD\subset \ker^m\Omega$. In particular, we will call the elements of $S(\Omega)$ as solutions, even though they may not be integrable. We will assume that $S(\Omega)$ is not empty.

\begin{remark}\label{rem:admisible} Most examples of (pre)multisymplectic systems impose extra conditions to the solutions. For instance, Lagrangian systems require that the solutions are holonomic. The sections or distributions that fulfill this extra conditions will be called \emph{admissible}. This implies that the precise definition of a solution may be different in every example. In sections \ref{sec:expsol} and \ref{sec:quotient} we will use the generic definition \ref{dfn:sol}. How to apply the results to systems with more restricted definitions of solutions is discussed in section \ref{sec:adm} and \ref{sec:applications}.
\end{remark}

\section{Expanded solutions}\label{sec:expsol}
We want to classify the elements of $S(\Omega)$, which we will assume is not empty (if that is not the case, one may use a constraint algorithm \cite{de_leon_pre-multisymplectic_2005}). In order to do so, we will group them inside new distributions called expanded solutions.

\begin{definition} A generalized distribution $D$ is an \emph{expanded solution} if it is transverse and $\Lambda^m{D}\subset Ker^m\Omega$.
\end{definition}

Since any transverse distribution has a regular subdistribution of rank $m$ which is also transverse, any expanded solution contains a solution. In particular, a trivial expanded solution (that is, an expanded solution with constant rank $m$) is an element of $S(\Omega)$.

\begin{definition}
Two expanded solutions $D_1$ and $D_2$ are \emph{kernel related} if $D_1+D_2$ is an expanded solution. It will be denoted by $D_1\sim_\Omega D_2$.
\end{definition}

In general, $\sim_\omega$ is not an equivalent relation. 

\begin{example}
Consider the fibre bundle $\pi:\mathbb{R}^8 \rightarrow \mathbb{R}^2$, with coordinates such that $\pi(x,y,q,p^x,p^y,u,v,w)=(x,y)$, and the premultisymplectic form
$$
\Omega=\d q\wedge\d p^x\wedge \d y-\d q\wedge\d p^y\wedge \d x+\d q\wedge\d u\wedge \d w\,.
$$
The following distributions with rank $3$ are expanded solutions:
\begin{align*}
    D_1=\left<\derpar{}{x},\derpar{}{y},\derpar{}{u}\right>\,;\quad  D_2=\left<\derpar{}{x},\derpar{}{y},\derpar{}{v}\right>\,;\quad  D_3=\left<\derpar{}{x},\derpar{}{y},\derpar{}{w}\right>\,.
\end{align*}
Moreover, $D_1+D_2$ and $D_2+D_3$ are also expanded solutions, thus $D_1\sim_\Omega D_2$ and $D_2\sim D_3$. Nevertheless, $D_1+D_3$ is not an expanded solution because $i{\left(\sfrac{\partial}{\partial u}\right)}i{\left(\sfrac{\partial}{\partial w}\right)}\Omega\neq 0$. 
\end{example}

As we will see in \ref{prop:equivrel},  if $\Omega$ is variational then $\sim_\Omega$ is an equivalence relation.

\subsection{Characterization with kernel symmetries}

There are several definitions of ''gauge" vector fields. In gauge theory, gauge vector fields represent infinitesimal transformation of the coordinates of the system (but not the space-time) which leave the system invariant. In multisymplectic systems this vectors fields leave the Lagrangian and the multisymplectic form invariant (see the example about electromagnetism \ref{app:EM}, or the case of general relativity \cite{Gaset2018}). Therefore, they are Lagrangian symmetries \cite{echeverria-enriquez_geometry_2000} and, in general, they are not in the kernel of the premultisymplectic form. 

In the Gotay-Nester algorithm \cite{gotay_presymplectic_1979}, every vector field in the kernel of the presymplectic form is a gauge vector field. Moreover, there are ''secundary" gauge vector fields constructed using the solution. 

Finally, in \cite{gaset_variational_2016} we defined the vector fields of the kernel of a premultisymplectic form as ''geometric" gauge vector fields. The name came from the Gotay\&Nester's definition, but we added the ''geometric" adjective to differentiate them from the usual gauge vector fields used in theoretical physics. Incidentally, in \cite{dadhich_equivalence_2012} the authors refer as gauge freedom a transformation which turns out to be described by a geometric gauge vector field \cite{gaset_new_2019}.

All these vector fields encapsulate the idea of redundant information that the word ''gauge" usually refers to. Nevertheless, having the same name to every object is impractical when we want to study the multiple solutions of a multisymplectic system. Consequently, in this work we will call the geometric gauge vector field as kernel vector fields, which is a more descriptive name.

\begin{definition}Given a system $(J,\Omega)$, the \emph{kernel distribution} is $K:=Ker^1\Omega\cap V(\pi)$.

A kernel (or geometric gauge) vector field is a smooth section of $K$,  which we will denote as $\Gamma(K)$.
\end{definition}


 


Kernel vector fields lead to expanded solutions:

\begin{proposition}\label{prop:kernel->ext}
Given  system $(J,\Omega)$, if there exists a section $g:J\rightarrow K$ (which we will not assume smooth), then for every expanded solution $D$ of $\Omega$ , $D+\!<\!g\!>$ is an expanded solution.
\end{proposition}
\begin{proof}
Since $D$ is transverse, $D+\!<\!g\!>$ is also transverse.
At every point $p\in J$ an element of ${\bf X}\in\Lambda^m(D_p+\!<\!g_p\!>)$ can be written as:
$${\bf X}=u_1\wedge\dots\wedge u_m+\sum_{i=1}^mg_p\wedge v_1\dots\wedge v_{m-1},$$

where $u_1,\dots,u_m,v_1,\dots,v_{m-1}\in D_p$. Using that $g_p\in K_p$ and $\Lambda^mD\subset Ker^m\Omega$, we derive that ${\bf X}\in Ker^m\Omega$. Therefore $D+\!<\!g\!>$ is an expanded solution.
\end{proof}

In particular, given a a solution of $\Omega$ we can use a section of $K$ to construct an extended solution. Another way of using this result is considering an open set where the solution is generated by a set of vector fields and a smooth section of $K$. Then, the addition of the section to any of the vector field generates a new distribution which is also a solution of $\Omega$.

The converse does not hold in general, that is, the existence of expanded solutions does not implies the existence of kernel vectors (see example \ref{ex:novariationalexpandsol}). Nevertheless, the equivalence holds for a special class of (pre)multisymplectic forms.

\begin{lemma}\label{lem:tech1} Consider a variational $m+1$ form $\omega$ and a vertical vector at a point $p\in J$,  $Y\in V_p(\pi)$. If there exists a transverse distribution $D$ of rank $m$, generated in $p\in J$ by $\{v_i\}_{i\in I}$, such that 
$$i\left({\bigwedge_{i\in L}v_i}\right)i({Y})\omega_p=0 \quad \text{ for every } L\subset I, |L|=m-1\,,$$
then $i_{Y}\omega_p=0$.
\end{lemma}
\begin{proof} $i_{Y}\omega_p \in \bigwedge^mT_p^*J$ and is equal to $0$ if it vanishes by the action of all elements of $\bigwedge^mT_pJ$. Since $D$ is transverse and has rank $m$, we can consider the decomposition $D+ V(\pi)=TJ$.  Therefore, any element of $\bigwedge^mT_pJ$ is a linear combination of elements of the from $\mathbf{X}=z_1\wedge \cdots\wedge z_k\wedge v_1\wedge\cdots\wedge v_r$, where $k+r=m$ and  $z_1,\dots, z_k\in V_p(\pi)$ and $v_1,\dots v_r\in D_p$. If $k>1$, then $i_\mathbf{X}(i_{Y}\omega_p)=0$ because $\omega$ is variational ($\mathbf{X}\wedge Y$ contains three or more vertical vectors or it is null).  If $k=1$:
$$i_\mathbf{X}(i_{Y}\omega)=i_{Y\wedge z_1\wedge v_1\wedge\cdots\wedge v_r}\omega=(-1)^ri_{z_1}(i_{Y\wedge v_1\wedge\cdots\wedge v_r}\omega_p)=0$$
for hypothesis. For $k=0$ we can use the same argument.
As a consequence, $i_Y\omega=0$
\end{proof}
\begin{theorem}\label{theo:ext->kernel}
If $\Omega$ is variational, every expanded solution $D$ can be decomposed as
$$D=H+K_D,$$
where $H$ is distribution solution of $(J,\Omega)$ and $K_D$ is a subdistribution of $K$, the kernel distribution.
\end{theorem}
\begin{proof}
Since it is an expanded solution, it has a solution, which can be choose to define $H$. Therefore, there exists a generalized distribution $B\subset V(\pi)$ sucht that $D=H+ B$. At a point $p\in J$ suppose that $H_p$ is generated by $\{X_i\}$, $i=1,\dots,m$. For any element $Y\in B_p$, we have
$$i\left({\bigwedge_{i\in L}X_i}\right)i({Y})\Omega_p=0 \quad \text{ for every } L\subset I, |L|=m-1\,,$$

because $\Lambda^m{D}\subset K^m$. Recalling that $\Omega$ is variational and using lemma \ref{lem:tech1}, we derive $i_Y\Omega_p=0$. As this holds for every point $p$, we conclude that $B\subset K$.
\end{proof}

The following example shows that being variational is not a necessary condition.
\begin{example}
Consider the fibre bundle $\pi:\mathbb{R}^5 \rightarrow \mathbb{R}^2$, with coordinates such that $\pi(x,y,q,p^x,p^y)=(x,y)$, and the multisymplectic form $
\Omega=\d q\wedge\d p^x\wedge \d y-\d q\wedge\d p^y\wedge \d x+\d q\wedge\d p^x\wedge \d p^y\,.
$ $\Omega$ is not variational, but it has no non-trivial expanded solutions, thus \ref{theo:ext->kernel} holds.
\end{example}

\begin{corollary}
Consider a system $(J,\Omega)$ where $\Omega$ is variational. $\Omega$ has a non trivial expanded solution  if, and only if,  $K \neq \{0\}$.
\end{corollary}

Finally, as a consequence of theorem \ref{theo:ext->kernel}, the induced relation $\sim_\Omega$ is indeed and equivalence relation when $\Omega$ is variational.

\begin{proposition}\label{prop:equivrel}
If $\Omega$ is variational, then $\sim_\Omega$ is an equivalence relation.
\end{proposition}
\begin{proof}
The relation is reflexive and symmetric by the properties of the sum of vector spaces. Consider three expanded solutions $D_1$, $D_2$ and $D_3$, such that $D_1\sim_\Omega D_2$ and $D_2\sim_\Omega D_3$.  By theorem \ref{theo:ext->kernel} and lemma \ref{lem:descomposa}, we can consider the decomposition:
\begin{align*}
    &D_1=H_1+B_1\,;&& D_2=H_2+B_2\,;&& D_3=H_3+B_3\,;
    \\
   & D_{1}+D_2=H_1+B_{12}\,;&& D_{2}+D_3=H_2+B_{23}\,;&& D_{1}+D_3=H_1+B_{13}\,.
\end{align*}

The distributions $H_1$, $H_2$ and $H_3$ are transverse with rank $m$; $B_1$, $B_2$, $B_3$, $B_{12}$, $B_{23}\subset K$ and $B_{13}\subset V(\pi)$. $D_{1}+D_3$ contains $H_1$, which is transverse and has rank $m$, thus $D_{1}+D_3$ is transverse. Moreover, since $H_3\subset H_2+B_{23}\subset H_1+B_{12}+B_{23}$, we have that:
$$
B_{13}\subset H_1+B_1+H_3+B_3\subset H_1+B_1+B_3 +B_{23}+B_{12}\,.
$$
Any element  of $B_{13}$ can be written as a sum of elements of the right hand side but,  since $B_{13}\subset V(\pi)$, the contribution of $H_1$ must be the vector $0$. Therefore, $B_{13}\subset B_1+B_3+B_{23}+B_{12}\subset K$. By proposition \ref{prop:kernel->ext}, $D_1+D_3$ is an expanded solution. 
\end{proof}

The condition of being variational is a sufficient condition but not necessary:

\begin{example}\label{ex:novariationalexpandsol}
Consider the fibre bundle $\pi:\mathbb{R}^6 \rightarrow \mathbb{R}^2$, with coordinates such that $\pi(x,y,q,p^x,p^y,u)=(x,y)$, and the multisymplectic form
$$
\Omega=\d q\wedge\d p^x\wedge \d y-\d q\wedge\d p^y\wedge \d x+\d u\wedge\d p^x\wedge \d p^y\,.
$$
$\Omega$ is not variational and $K=\{0\}$. Nevertheless, we have a non-trivial expanded solution:
$$
D=\left<\derpar{}{x},\derpar{}{y},\derpar{}{u}\right>\,.
$$
$\Omega$ is not variational but $\sim_\Omega$ is an equivalence relation, thus being variational is a sufficient condition but not necessary.
\end{example}

\section{Quotient and reduction}\label{sec:quotient}

The equivalence relation presented in the previous section allows us to perform a quotient in the space of solutions of the system, as a distributions. Nevertheless, in application one prefer to take the quotient in the manifold or, at least, in the space of solutions as a section.

The strong kernel distribution $K$ is involutive thus it induces a folliation of $J$. Nevertheless, the set of leafs, denoted $\sfrac JK$, may not have good properties. From now on, we assume that $\sfrac JK$ has a differential structure such that the quotient manifold $\xi:J\rightarrow \sfrac JK$ is a submersion. Since $K$ is a subset of vertical vectors, there exists a unique projection $\pi_K: \sfrac JK\rightarrow M$ such that $\pi_K\circ\xi=\pi$. Namely, the following diagram commutes:

$$
\xymatrix{
& J\ar[rrd]^<(0.45){\xi}\ar[dd]^<(0.45){\pi} \ & \ &
\\
& \ & \ & \sfrac JK\ar[lld]^<(0.45){\pi_K}
\\
& M\ &  \ &
}
$$

We are interested in what extend we can relate solutions of the original system and the quoted. First of all, not all solutions of the original problem are projectable to the quoted system.
\begin{definition}
A distribution $D\subset TJ$ is $\xi$-projectable if, at every point $q\in \sfrac JK$ and for every pair of points in the fiber $p_1, p_2\in \xi^{-1}(q)$, 
$$
T_{p_1}\xi(D_{p_1})=T_{p_2}\xi(D_{p_2}).
$$
\end{definition}

If a distribution is projectable, we can associate a corresponding distribution in $\sfrac JK$ by choosing a section $\beta$ of $\xi$ (a representative for every class):
$$
D_K\equiv\bigcup_{q\in \sfrac JK}T_{\beta(q)}\xi(D_{\beta(q)})
=T\xi(D)$$

\begin{lemma}
The associated distribution to a projectable $\pi$-transverse distribution is $\pi_K$-transverse
\end{lemma}
\begin{proof}
$$
T_q\pi_K({D_K}_q)=T_q\pi_K(T_{\beta(q)}\xi(D_{\beta(q)}))=T_{\beta(q)}(\pi_K\circ\xi)(D_{\beta(q)})=T_{\beta(q)}\pi\left( D_{\beta(q)}\right)=T_{\pi(\beta(q))}M=T_{\pi_K(q)}M.
$$
\end{proof}

\begin{lemma}
Let $D_1$ and $D_2$ be two transverse generalized distributions such that $D_1+D_2=H+B$, where $H$ is a transverse distribution and $B\subset K$. Then, $D_1^K=D_2^K$.
\end{lemma}
\begin{proof}
Any vector of $D_1$ or  $D_2$ can be written as an element of $H$ plus elements of $K$, then $T\xi (D_i)\subset T\xi (H)$, for $i=1,2$. Since $D_1$ and $D_2$ are transverse and $K$ is vertical, the projection have the same dimension, thus $T\xi (D_1)=T\xi (H)=T\xi (D_2)$.
\end{proof}

\subsection{Kernel-related sections}

\begin{definition}
Two sections are \emph{kernel-related} if $\phi_1(p) \sim \phi_2(p)$ for every $p\in M$. Equivalently, $\xi\circ\phi_1=\xi\circ\phi_2$.
\end{definition}
This relations is an equivalence relation. This definition is perhaps the closest to the physical notion of kernel-related solutions. 

\begin{proposition}
Consider two integrable distributions $X_1$, $X_2$ solution of $\Omega$ and $\xi$-projectable. They are kernel-related if, and only if, at every point $p\in J$ the corresponding integrable sections $\phi_1$ and $\phi_2$ are kernel related.
\end{proposition}
\begin{proof}
Suppose $X_1$ and $X_2$ are kernel related, with $D\equiv X_1+X_2=H+K'$. $Y$, the distribution $\xi$-related to $D$, is transverse and has dimension $m$ thus it is the same as the distribution related to $X_1$ and $X_2$. This implies that $\xi_*\phi_1$ and  $\xi_*\phi_2$ are integral section of $Y$. They coincide at $\xi(p)$, therefore $\xi_*\phi_1=\xi_*\phi_2$ because $Y$ is involutive. For any point $x\in M$, $\xi_*\phi_1(x)=\xi_*\phi_2(x)$ thus $\phi_1(x)\sim\phi_2(x)$ and they are kernel related.

Conversely, if $\phi_1$ and $\phi_2$ are kernel related they project to a unique section $\xi_*\phi_1=\xi_*\phi_2$. Then, if $x=\pi(p)$:

$$
T_p\xi X^1_p=T_p\xi T_x\phi_1 T_xM=T_x(\xi\circ \phi_1)T_xM=T_x(\xi\circ \phi_2)T_xM=T_p\xi T_x\phi_2 T_xM=T_p\xi X^2_p.
$$
Therefore $D_p^1+D_p^2=D_p^1+K'_p$, where $K'_p\subset K_p$ and they are kernel-related.
\end{proof}

\subsection{Reduced system}

The multisymplectic form $\Omega$ is closed, then, for any strong kernel vector field $X\in\Gamma(K)=\mathfrak{X}^{V(\xi)}J$:
$$
\Lie_X\Omega=\d \inn(X)\Omega= \inn(X) \d \Omega=0.
$$
Therefore, $\Omega$ is $\xi$-projectable. In this case there exists a form $\Omega_K\in\Omega^{m+1}(\sfrac JK)$ such that, 
$$
\xi^*\Omega_K=\Omega.
$$
This form can be constructed as $\Omega_K=\beta^*\Omega$, for any $\beta$ a section of $\xi$. Since $\Omega$ is $\xi$-projectable, it does not depend on the section chosen. The quoted system is $(\sfrac JK,\Omega_K)$. 
\begin{lemma}
$\Omega_K$ is a multisymplectic form, that is, it is closed and $1$-nondegenerate.
\end{lemma}
\begin{proof} Independent of the section $\beta$ chosen, 
$$
\d\Omega_K=\d\beta^*\Omega=\beta^*\d\Omega=0 \ .
$$
Now supose $\Omega_K$ is $1$-degenerate. That means there exists a vector field $Y\in\mathfrak{X}(\sfrac JK)$, different form $0$, such that $\inn(Y)\Omega_K=0$. But taking the pull-back by $\xi$, and considering that $X$ is $\xi$-related to $Y$:
$$
0=\xi^*\inn(Y)\Omega_K=\inn(X)\Omega.
$$
Then $X$ is a kernel vector field of $\Omega$, and $Y\circ\xi=T\xi(X)=0$, which is a contradiction.
\end{proof}

\begin{proposition} \label{prop:kernel}
\begin{enumerate}
    \item  If the section $\psi$ is a solution of the system $(J,\Omega)$, then $\xi\circ \psi$ is a solution of the system  $(\sfrac JK,\Omega_K)$
    \item If the section $\phi$ is a solution of the system $(\sfrac JK,\Omega_K)$, then $\beta\circ \phi$ is a solution of the system  $(J,\Omega)$ for every section $\beta$.
\end{enumerate}
\end{proposition}
\begin{proof}
\begin{enumerate}
    \item Since $\pi_K\circ\xi=\pi$, we have that $\pi_K\circ\xi\circ\psi=\pi\circ\psi=Id_M$. As a consequence $\xi$ restricts to a diffeomorphism  between $\Ima \phi$ and $\Ima (\xi\circ\phi)$. Fixed a point $q\in \Ima (\xi\circ\phi)$ we have a unique antimage $\xi(p)=q$, and for every $Y\in T_q\sfrac JK$, there exists a vector $v\in T_pJ$ such that $T_p\xi(v)=Y$. Then:
$$(\xi\circ\psi)^*\inn(Y)\Omega_K|_q=\psi^*\left(\xi^*\inn(T_p\xi(v))\Omega_K|_q\right)=\psi^*\inn(v)\Omega|_p=0\ .$$
    \item Conversely, in the points $p\in \Ima\beta$, the section induces the splitting $T_pJ=T_{\xi(p)}\beta(T\sfrac JK)+K_p$. Therefore, all vectors $v\in T_pJ$ can be written as $v=u+g$, where $u\in T_{\xi(p)}\beta (T\sfrac JK)$ and $g\in K_p$. Notice that only it is required to check the field equations at the points of the image of the section we are testing, so consider $p\in\Ima(\beta\circ\psi)$:
$$(\beta\circ\psi)^*\inn(v)\Omega|_p=\psi^*(\beta^*\inn(u+g)\Omega|_p)=\psi^*\inn(\beta_*v)\Omega_K|_{\xi(p)}=0\ .$$
We have used that $i(g)\Omega|_p=0$ because $g$ is a kernel vector.
\end{enumerate}
\end{proof}

Now $Ker^1(\Omega_K)_p\cap Ker(T_p\pi)=\emptyset$. If $\Omega$ is variational, from proposition 2 we know that there is no expanded solution. As a consequence, maybe $\Omega$ has several distributions solution, but they are not kernel-related. We manage to eliminate some of the multiple solutions of the system in a way similar of the mechanical systems. The discussion about which solutions are or are not physical relevant is beyond this work.

\section{Admissible solutions}\label{sec:adm}

In most applications, not every transverse distribution in $Ker^m\Omega$ is a valid solution, as some extra conditions are required. For instance, it is habitual to require that the distribution is integrable, and in Lagrangian system, the solutions are required to be holonomic. A distribution or section which are solutions and also satisfy these extra conditions are labelled \textit{admisible} (see remark \ref{rem:admisible}). 

The extra conditions are implemented with different structures in each model, which make it difficult to present a general theory. Unfortunately, in most applications it is required to understand the interplay between kernel related solutions and admissible solutions. In this section and in section \ref{sec:applications} we will discuss some common situations.

The general procedure we recommend is to study which solutions are kernel-related first, and then to impose admissibility. This is the usual method to deal with integrability in multisymplectic systems. This is because, generally, an admissible solution will be kernel-related to a non-admissible one.  With this approach one may perform the quotient as in section \ref{sec:quotient}, although the admissibility condition do not tend to project to the quotient. This is relevant if one tries to use proposition \ref{prop:kernel} to recover solutions from the quoted system. It would be needed to check which recovered solutions are actually admissible. An illustrative example is the case of Electromagnetism, where the holonomy and integrability conditions affect expanded solutions. 

There are special situations, like the Metric-Affine model for General Relativity \cite{gaset_new_2019}, where the kernel distribution is compatible with the admissibility conditions. In this case, the quoted system is an accurate representation of the whole system.

\subsection{Weak kernel distribution}

In some situations, there are vector fields $Y$ such that $i(Y)\Omega\neq 0$, but $\psi^*\inn(v)\Omega=0$ for any admissible section. In this section we present a reduction of the system by this kind of vector fields.

\begin{definition}
\begin{itemize}
    \item The \emph{weak kernel distribution} of the system $(J,\Omega)$ is 
    $$
    K_{w}\equiv\bigcup_{p\in J_f}\{v\in V_p |\, \psi^*\inn(v)\Omega|_{\psi^{-1}(p)}=0\ ,\, \text{ for all admisible sections } \psi:M\rightarrow J \text{ such that } p\in Im(\psi)\ \}.
    $$
    \item   A {\rm weak kernel vector field} of the system $(J,\Omega)$ is a vertical vector field such that
    $$
    \psi^*(\inn(Y)\Omega)=0\ ,\quad \text{ for all admisible sections } \psi:M\rightarrow J \ .
    $$
\end{itemize}
\end{definition}

\begin{example}
The Higher order Lagrangian field theories are defined over the jet bundle $J^{1}\pi$ of a fiber bundle $\pi:E\rightarrow M$. In these systems it is required that the admissible sections are holonomic. This can be encoded using $\mathcal{C}$, the Cartan codistribution, generated by the forms:
$$
\theta^\alpha=\d u^\alpha- u^\alpha_{i}\d x^i\,,
$$
where $(x^i,u^\alpha_i)$ are adapted coordinates in $J^{1}\pi$. Then, a section $\phi:M\rightarrow J^{1}\pi$ is holonomic (admissible) if $\phi^*\theta=0$, for all $\theta\in \mathcal{C}$ (see \cite{saunders_geometry_1989} for more details).

Given a regular Lagrangian $L: J^1\bar{\pi}\rightarrow \mathbb{R}$, the corresponding multisymplectic form has the local expression:
$$
\Omega=\d\left(u^\alpha_i\frac{\partial L}{\partial u^\alpha_i}-L\right)\wedge\d^mx-\d\frac{\partial L}{\partial u^\alpha_i}\wedge\d u^\alpha\wedge \d x^{m-1}_i\,.
$$
If we contract by a vector field tangent to the fibers $ J^1\bar{\pi}\rightarrow E$, we have that:
$$
i\left(\sfrac{\partial}{\partial u^\beta_j}\right)\Omega=u^\alpha_i\frac{\partial^2 L}{\partial u^\alpha_i\partial u^\beta_j}\d^mx-\frac{\partial^2 L}{\partial u^\alpha_i\partial u^\beta_j}\d u^\alpha\wedge \d x^{m-1}_i=-\frac{\partial^2 L}{\partial u^\alpha_i\partial u^\beta_j}\theta^\alpha\wedge \d x^{m-1}_i\,.
$$
Therefore, $\psi^*(\inn\left(\sfrac{\partial}{\partial u^\beta_j}\right)\Omega)=0$ for any holonomic solution. Namely, the vector fields in the directions of the velocities are weak kernel vector fields. This is also the case for higher order regular Lagrangians: the vertical vector fields of the fiber bundle $J^{2k-1}\pi\rightarrow E$ are weak kernel vector fields.
\end{example}

They are relevant when analyzing the condition set and the conserved quantities. For instance, the difficulties on defining the Hamilton-Cartan form for the Hamiltonian formalism for higher-order field theories is a consequence that the symetrization of the momenta leads only to weak kernel vector fields.

Consider an involutive subdistribution of  the weak kernel distribution $K'_w\subset K_w$ and assume that $\sfrac JK$ has a differential structure such that the quotient manifold $\xi:J\rightarrow \sfrac JK$ is a submersion. Given a section (kernel fixing) of $\xi$, we can define $\Omega_\beta \equiv\beta^*\Omega$, but it depends on the section because, in general, $\Omega$ is not constant in $\xi^{-1}(p)$ for some $p\in J/K'_w$.

\begin{proposition} \label{prop:kernel}
Consider an smooth section $\beta$ of $\xi$. 
\begin{enumerate}
    \item 
    If the section $\psi$ is a solution of the system $(J,\Omega)$ such that $\beta\circ\xi\circ\psi=\psi$, then $\xi\circ \psi$ is a solution of the system  $(\sfrac{J}{K'_w},\Omega_\beta)$
    \item
    If the section $\phi$ is a solution of the system $(\sfrac {J}{K'_w},\Omega_\beta)$, and $\beta\circ \phi$ is admisible, then $\beta\circ \phi$ is a solution of the system  $(J,\Omega)$. 
\end{enumerate}
\end{proposition}
\begin{proof}
\begin{enumerate}
    \item For every $Y\in \Gamma (V_s)$, we have that $\beta_*Y\in \Gamma(V)$; then:
$$(\xi\circ\psi)^*\inn(Y)\Omega_s=(\xi\circ\psi)^*\inn(Y)\beta^*\Omega=(\beta\circ\xi\circ\psi)^*\inn(\beta_*Y)\Omega=\psi^*\inn(\beta_*Y)\Omega=0\ .$$
\item Since $\beta$ is a section of $\xi$ , we have the decomposition on the image of $\beta$: $V=Im(T\beta)+ K'_w$. Then, every
$X\in\Gamma(V)$ can be written as $X=Y+Z$ with $X\in\Gamma(Im(T\beta))$ and $Z\in\Gamma(K'_w)$:
$$(\beta\circ\phi)^*\inn(X)\Omega=(\beta\circ\phi)^*\inn(Y)\Omega+(\beta\circ\phi)^*\inn(Z)\Omega=(\phi)^*\inn(\beta^*Y)\beta^*\Omega=\phi^*\inn(\beta^*Y)\Omega_\beta=0.$$
The term with $Z$ vanishes because $\beta\circ\phi$ is admissible and $Z$ is a weak kernel vector field.
\end{enumerate}
\end{proof}

\section{Distributions and Mechanical systems}


Consider the fiber bundle $\pi: E\rightarrow\mathbb{R}$ with dim $E=n+1$. The mechanical information is encoded in a closed 2-form $\Omega$. The problem consists on finding a distributions $D$ such that $i_X\Omega=0$ for every section $X\in\Gamma(D)$. Moreover, it has to be transverse: $T_p\pi (\mathcal{D})=T_{\pi(p)}\mathbb{R}$, integrable and $dim D=1$. \footnote{Usually there are more features, like holonomy. For now on, I will not consider them.} 

This is the optimal physical situation, because if such a distribution can be found, then there exists a folliation of $E$ of dimension 1, where every leaf is (the image of) a section of $\pi$. 

In this work I will start with the distribution $K:=\cup_{p\in E} Ker(\Omega_p)$ and I will try to extract $D$ from it. First of all, I will assume that $K$ is transverse \footnote{It is equivalent to have at least one vector field which is transverse. If not, we should use the constraint algorithm} and has constant dimension. By definition $i_X\Omega=0$ for every $X\in\Gamma(K)$. Moreover, if $X,Y\in \Gamma(K)$, then $[X,Y]\in \Gamma(K)$. Indeed:

$$i_{[X,Y]}\Omega=\Lie_Xi_Y\Omega-i_Y\Lie_X\Omega=-i_Y(i_X\d\Omega+\d (i_X\Omega))=0.$$

So $K$ is involutive, and therefore integrable. If $dimK=k=1$, we had just found the solution: $D=K$\footnote{Unicitat?}. What we are interested is the case $k>1$. Notice that in this situations we have a folliation of dimension $k$. We assume that the extra $k-1$ degrees of freedom has no physical meaning. Therefore we should make something with them. What I will do is quoting the leaves to curves in a canonical way, thanks to the transverse condition.

Consider $G:=K\cap V(\pi)$ (vertical vectors), resulting $\Gamma(G)=\{g\in\mathfrak{X}(E)|i_g\Omega=0,\;i_g\pi^*\d t=0\}=$. As we will soon see, these vector fields contains the non-physical information, and I will call them kernel vector fields. Clearly, $G\subset K$ and the next lemma gives us its dimension:

\begin{lemma}
$G$ has codimension $1$ in $K$.
\end{lemma}
\begin{proof}
Since $G\subset K$, it has finite dimension. Lets say $dim G=r$.
Consider a fixed transverse vector field $X\in\Gamma(K)$. For every point $p\in E$ we have that $i_X\pi^*\d t|_p=x\neq 0$. We also have a base of $G_p=<g^1,g^2,\dots,g^r>$. I will proof that $\{g^1,g^2,\dots,g^r,X_p\}$ is a base of $K_p$.

Given a vector $Y\in K_p$, define $y:=i_Y(\pi^*\d t)|_p$. We can have that $y=0$, then $y\in G_p$. Otherwise $y\neq 0$, then $Y-\frac yx X_p\in G_p$ because 
$$i_{(Y-\frac yx X_p)}(\pi^*\d t)|_p=i_{Y}(\pi^*\d t)|_p-\frac yx i_{X_p}(\pi^*\d t)|_p=y-y=0.$$ 
Therefore, $<g^1,g^2,\dots,g^r,X_p>$ span $K_p$. Moreover, $X_p\not\in G_p$ because it is transverse and $X_p\not\in V(\pi)$.
\end{proof}

In other words, the last lemma proved that $<\!X\!>\cup \;G$ span $K$. (Maybe more correct : $K=\cup_{p\in E}<X_p>+G_p$) 

As it is easy shown, $G$ is an involutive distribution\footnote{Potser no tant evident, pero no es complicat}, therefore induces a folliation. Is time to take the quotient. $\sfrac EG=\{\overline{p}\}$ is the set of leaves, with the natural application $\psi: E \hookrightarrow \sfrac EG$. Take $\phi$, a section of $\psi$. The projection is $\pi_G=\phi^*(\pi):\sfrac EG\rightarrow \mathbb{R}$. It is well defined because $G$ is vertical. The pull-back 2-form is $\Omega_G:=\phi^*(\Omega)$. For $v,u\in T_{\overline{p}}  \sfrac EG$, $\phi^*(\Omega)|_{\overline{p}}(u,v)=\Omega_{\phi(\overline{p})}(T_{\overline{p}}\phi(u),T_{\overline{p}}\phi(v))$.
Since $G\subset K$, it doesn't depend on the section choosen and $\Omega_G$ is well defined.\footnote{Es pot explicar millor, pero he de buscar algun lloc on expliquin detalladament el quocient de varietats}

Now $\cup_{p\in E} Ker({\Omega_G}_p)$ has dimension 1, and it gives a folliation of $\sfrac EG$ (which now has dimension $n+1-(k-1)$) where the leafs are curves, which are (the image of) sections of  $\pi$.

\section{Applications} \label{sec:applications}


\subsection{Electromagnetism}\label{app:EM}

Let $M=\mathbb{R}^4$ represent space-time, $P \rightarrow M$ the principle bundle with structure group $U(1)$ and $\pi:C\rightarrow M$ the associated bundle of connections (see \cite{castrillon_lopez_geometry_2001} for more details).
The Electromagnetic Lagrangian
\[
L =  -\frac{1}{4\mu_0}\eta^{\alpha\mu}\eta^{\beta\nu}F_{\mu\nu}F_{\alpha\beta}\ ,
\]
where $F_{\mu\nu}=A_{\nu,\mu}-A_{\mu,\nu}$ is the electromagnetic tensor field, $\gamma_\alpha\in C^\infty(M)$ are smooth functions (for $0\leq\alpha\leq 3$), $\eta^{\mu\nu}$ is the Minkowski metric on $M$ and $\mu_0$ is a constant. The Lagrangian energy is
\[
 E_\mathcal{L} = A_{\mu,\alpha}\frac{\partial L}{\partial A_{\mu,\alpha}}-L = \frac{1}{\mu_0}\eta^{\mu\nu}\eta^{\alpha\beta}A_{\mu,\alpha}F_{\nu\beta}+\frac{1}{4\mu_0}\eta^{\mu\nu}\eta^{\alpha\beta}F_{\beta\nu}F_{\alpha\mu}\ ,
\]
and the premultisymplectic form is
\[
\Omega=\d E_{\mathcal{L}}\wedge \d ^4x-\frac{1}{\mu_0}\left(\eta^{\alpha\nu}\eta^{\mu\beta}-\eta^{\alpha\beta}\eta^{\mu\nu}\right)\d A_{\beta,\nu}\wedge\d A_\alpha\wedge\d ^3x_\mu\ .
\]
The kernel distribution is $K=\left<\sfrac{\partial}{\partial A_{\alpha,\mu}}+\sfrac{\partial}{\partial A_{\mu,\alpha}}\right>$, which is involutive. We will now compute explicitly the solutions of the system to see how the kernel distribution helps to describe them.

Consider a transverse distribution $D=\left.\left<X_\mu\right>\right|_{\mu=0,1,2,3}$ with local expression
$$
X_\mu=\frac{\partial}{\partial x^\mu}+G_{\alpha,\mu}\frac{\partial}{\partial A_\alpha}+G_{\alpha\nu,\mu}\frac{\partial}{\partial A_{\alpha,\nu}}\,,
$$
for some functions $G_{\alpha,\mu}, G_{\alpha\nu,\mu}\in C^\infty( J^1\bar{\pi})$.
$D$ is a solution of the system \ref{dfn:sol}, without imposing integrability nor holonomy, if
$$
G_{\alpha,\mu}-G_{\mu,\alpha}=A_{\alpha,\mu}-A_{\mu,\alpha}\,,\quad \eta^{\nu\sigma}\eta^{\mu\tau}\left(G_{\tau\sigma,\nu}-G_{\sigma\tau,\nu}\right)=0\,.
$$
These equations have multiple solutions, which can be made more explicit by considering functions $R_{\alpha,\mu}, S_{\alpha\nu,\mu}, T_{\alpha\nu,\mu}\in C^\infty( J^1\bar{\pi})$ such that:
\begin{align*}
    R_{\alpha,\mu}-R_{\mu,\alpha}=0\,;\quad S_{\alpha\nu,\mu}-S_{\nu\alpha,\mu}=0\,;\quad
    T_{\alpha\nu,\mu}+T_{\nu\alpha,\mu}=0\,;\quad \eta^{\nu\mu}T_{\alpha\nu,\mu}=0\,.
\end{align*}
Then, the functions of the local expression of the distribution can be written as 
$$
G_{\alpha,\mu}=A_{\alpha,\mu}+R_{\alpha,\mu}\,,\quad G_{\alpha\nu,\mu}=T_{\alpha\nu,\mu}+S_{\alpha\nu,\mu}\,.
$$
Therefore, a distribution solution of the system has the local expression:
$$
X_\mu=\frac{\partial}{\partial x^\mu}+\left(A_{\alpha,\mu}+R_{\alpha,\mu}\right)\frac{\partial}{\partial A_\alpha}+\left(T_{\alpha\nu,\mu}+S_{\alpha\nu,\mu}\right)\frac{\partial}{\partial A_{\alpha,\nu}}\,.
$$
Every choice of $R_{\alpha,\mu}$, $S_{\alpha\nu,\mu}$ and $T_{\alpha\nu,\mu}$ produces a new solution. If we add a kernel vector field to any of the $X_\mu$ of a solution $D$, we generate a new solution with the same $R_{\alpha,\mu}$ and $T_{\alpha\nu,\mu}$ but different $S_{\alpha\nu,\mu}$. Indeed, two solutions are kernel-related if they have the same $R_{\alpha,\mu}$ and $T_{\alpha\nu,\mu}$. Thus, the kernel distribution describes the multiplicity of solution characterised by the functions $S_{\alpha\nu,\mu}$, but cannot describe the  multiplicity of solution characterised by the other functions.

This is a Lagrangian field theory, therefore, in order to have physical meaning, we should only consider as admissible solutions distributions which are integrable and holonomic. In local coordinates, a necessary condition for a distribution to be holonomic is that $G_{\alpha,\mu}=A_{\alpha,\mu}$. Therefore, $D$ is holonomic as long as $R_{\alpha,\mu}=0$. This is a common situation in the (pre)multisympelctic formalism of singular Lagrangians, as only part of the holonomic conditions can be recovered from the equations.

The distribution $D$ is integrable if $[X_\mu,X_\nu]=0$ for any $0\leq\mu<\nu\leq3$. If the distriution is holonomic, this implies:
$$
G_{\alpha\nu,\mu}-G_{\alpha\mu,\nu}=0\,;\quad  X_\mu (G_{\alpha\rho,\mu})-X_\nu(G_{\alpha\rho,\nu})=0\,,
$$
These are non-trivial equations which imposes relations between both families of functions. This implies that an equivalent class of solutions may contain some distributions that are integrable and others that are not.

Finally, a section $\phi:M\rightarrow  J^1\bar{\pi}$ is holonomic if it can be written as $j^1\psi$ for a section $\psi:M\rightarrow E$. The kernel distribution is vertical with respect to the projection $\pi^1: J^1\bar{\pi}\rightarrow E$, therefore, there exists a projection $\pi^1_K: \sfrac { J^1\bar{\pi}}{K}\rightarrow E$ such that $\pi^1_K\circ\xi=\pi^1$. Namely, the following diagram commutes:

$$
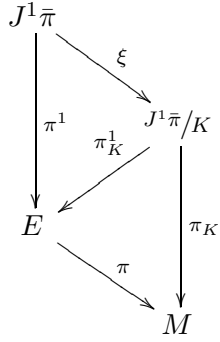
\xymatrix{
&  J^1\bar{\pi}\ar[rd]^<(0.45){\xi}\ar[dd]^<(0.45){\pi^1} \ & \ &
\\
& \ & \sfrac { J^1\bar{\pi}}{K}\ar[ld]_<(0.2){\pi^1_K} \ar[dd]^<(0.45){\pi_K} \ & 
\\
& E\ar[rd]^<(0.45){\pi}\ &  \ &
\\
& \ &  M \ &
}
$$

If two sections $\phi_1=j^1\psi_1$ and $\phi_2=j^1\psi_2$ are kernel-related, then $\xi\circ\phi_1=\xi\circ\phi_2$. Projecting to $E$ we have that 
$$
\pi^1_K\circ\xi\circ\phi_1=\pi^1_K\circ\xi\circ\phi_2\Rightarrow \pi^1\circ\phi_1=\pi^1\circ\phi_2\Rightarrow \psi_1=\psi_2\,.
$$
Therefore, $\phi_1$ and $\phi_2$ are equal. This situation happens in general for higher order Lagrangian field theories: if the kernel distribution is vertical with respect to the projection $J^r\bar{\pi}\rightarrow E$, then only one representant of an equivalence relation of sections is holonomic. In other words, only those kernel vector fields which induces a non-trivial flow on $E$ persist after imposing the holonomic condition.

The classical gauge symmetries of electromagnetism is the invariance under the transformation $A_\alpha\rightarrow A_\alpha+\frac{\partial f}{\partial x^\alpha}$, for some $f\in C^\infty(M)$. The infinitessimal version of this transformation is the vector field
$$
Y_f=\frac{\partial f}{\partial x^\alpha}\frac{\partial}{\partial A_\alpha}+\frac{\partial^2 f}{\partial x^\mu\partial x^\alpha}\frac{\partial}{\partial A_{\alpha,\mu}}\,,
$$
which has been lifted to the first jet. $Y$ is not a kernel nor weak kernel vector field. But we have that $$
\mathcal{L}_{Y_f}L=0\,;\quad \mathcal{L}_{Y_f}\Omega=0\,.
$$
Therefore, the gauge transformation generate Lagrangian symmetries. 

\subsection{Metric-Affine Lagrangian}\label{app:MA}

The Metric-Affine formulation of General Relativity is an interesting example, as it exhibits most of the characteristics discussed previously. The proofs of some of the results involve long computations, which I will avoid to write here but they can be found in \cite{gaset_new_2019}.

As we have seen in the electromagnetic case, the holonomy condition makes it difficult to have a clear description of the multiplicity of solutions and they relation to the kernel of the premultisymplectic form. Therefore, we will consider the covariant Hamiltonian formalism, which we will construct from the Lagrangian formalism.

The configuration bundle of this system is the bundle $\pi\colon E\rightarrow M$,
where $M$ is a connected orientable 4-dimensional manifold representing space-time and
${\rm E}=\Sigma\times_MC(LM)$, where $\Sigma$ is the manifold of Lorentzian metrics on $M$
and $C(LM)$ is the bundle of connections on $M$;
that is, linear connections in $\Tan M$. We will use the following coordinates in $J^1\pi$: $(x^\mu,\,g_{\alpha\beta},\,\Gamma^\nu_{\lambda\gamma},\,g_{\alpha\beta,\mu},
\,\Gamma^\nu_{\lambda\gamma,\mu})$, where $g_{\alpha\beta}$ represents the components of the metric and $\Gamma^\nu_{\lambda\gamma}$ the components of the connection which, in general, is not the Levi-Civita connection of $g$. The Lagrangian for Metric-Affine  gravity is:

$$L=\sqrt{|{\rm det}(g)|}\,g^{\alpha\beta}R_{\alpha\beta}(\Gamma)\equiv
\sqrt{|{\rm det}(g)|}\,g^{\alpha\beta}(\Gamma^{\gamma}_{\beta\alpha,\gamma}-\Gamma^{\gamma}_{\gamma\alpha,\beta}+
\Gamma^{\gamma}_{\beta\alpha}\Gamma^{\sigma}_{\sigma\gamma}-
\Gamma^{\gamma}_{\beta\sigma}\Gamma^{\sigma}_{\gamma\alpha})\,.\ $$

The Hamiltonian formalism takes place on in the image of the Legendre transform $\mathcal{FL}:J^1\pi\rightarrow J^*\pi$ \cite{echeverria-enriquez_multimomentum_2000}. It turns out that $Im(\mathcal{FL})$ is diffeomorphic to $E$ \cite{gaset_new_2019}. Therefore, we will consider the fiber bundle $\pi\colon E\rightarrow M$, with coordinates $(x^\mu,\,g_{\alpha\beta},\,\Gamma^\nu_{\lambda\gamma})$. The premultisymplectic form for the covariant Hamiltonian formalism is:

\beq\nonumber
\Omega_H=\d \left(\sqrt{|{\rm det}(g)|}g^{\alpha\beta}\left(\Gamma^\gamma_{\beta\sigma}\Gamma^\sigma_{\gamma\alpha}-\Gamma^\gamma_{\beta\alpha}\Gamma^{\sigma}_{\sigma\gamma}\right)\right)\wedge\d^4x -
\d\left(\sqrt{|{\rm det}(g)|}\left(\delta^\mu_\alpha g^{\beta\gamma}-\delta^\beta_\alpha g^{\mu\gamma}\right)\right)\wedge\d \Gamma^{\alpha}_{\beta\gamma}\wedge \d^3x_{\mu}  
\ .
\eeq

There is no transverse distribution which is solution of the system $(E,\Omega_H)$. Therefore, a constraint algorithm is performed in order to find a submanifold $E_f\subset  E$ where solutions exist and are tangent to $E_f$. This submanifold is described by the following constraints:
$$
t^\alpha_{\beta\gamma}\equiv T^\alpha_{\beta\gamma}-
\frac13\delta^\alpha_\beta T^\mu_{\mu\gamma}+
\frac13\delta^\alpha_\gamma T^\mu_{\mu\beta}=0 \,.
$$
 The torsion of the metric is defined as $T^\alpha_{\beta\gamma}\equiv\Gamma^\alpha_{[\beta\gamma]}=\Gamma^\alpha_{\beta\gamma}-\Gamma^\alpha_{\gamma\beta}$. The kernel distribution of $\Omega$ tangent to $E_f$ is
$$
K=
{\left<\delta^\alpha_\gamma\frac{\partial}{\partial \Gamma^\alpha_{\beta\gamma}}\right>} \,,
$$
where $\delta^\alpha_\delta$ is the identity tensor. Therefore, we expect that the solutions have a freedom at least in the direction for these vector fields. To check it, we will show the solutions explicitly. 

A transverse  distribution $D=\left.\left<X_\nu\right>\right|_{\nu=0,1,2,3}$ is solution of the system if
$$
X_\nu= \frac{\partial}{\partial x^\nu}+
\sum_{\sigma\leq\rho}\left(
g_{\sigma\lambda}\Gamma^\lambda_{\nu\rho}+g_{\rho\lambda}\Gamma^\lambda_{\nu\sigma}+
\frac{2}{3}g_{\sigma\rho}T^\lambda_{\lambda\nu}\right)
\frac{\partial}{\partial g_{\sigma\rho}}+\left
(\Gamma^\lambda_{\nu\gamma}\Gamma^\alpha_{\beta\lambda}+
C_{\beta\nu}\delta^\alpha_\gamma+K^\alpha_{\beta\gamma,\nu}\right)
\frac{\partial}{\partial \Gamma^\alpha_{\beta\gamma}}\,.
$$
for some functions 
$C_{\beta,\nu}, K^\alpha_{\beta\gamma,\mu}\in C^\infty(E_f)$ satisfying that
$$
K^\nu_{\nu\gamma\mu}=0\,; \quad
K^{\nu}_{\beta\gamma \nu}+K^\nu_{\gamma\beta \nu}=0 
\,; \quad K^\alpha_{[\beta\gamma],\mu}=-\frac13\delta^\alpha_{[\beta} K^\nu_{\gamma]\nu,\mu}-\Gamma^\lambda_{\mu[\gamma}\Gamma^\alpha_{\beta]\lambda}+\frac13\delta^\alpha_{[\beta}\Gamma^\lambda_{\mu\gamma]}\Gamma^\nu_{\nu\lambda}-\frac13\delta^\alpha_{[\beta}\Gamma^\lambda_{\mu\nu}\Gamma^\nu_{\gamma]\lambda}\,.
$$
Every distribution solution of the system is characterized by the choice of the functions $C_{\beta\mu}$ and $K^\alpha_{\beta\gamma,\nu}$. A class of kernel-related distributions is given by fixing only $K^\alpha_{\beta\gamma,\nu}$ because the freedom in choosing $C_{\beta\mu}$ is equivalent to adding a kernel-vector field.

A section $\phi:M\rightarrow E_f$ induces a metric and a connection, which, in general, it is not the Levi-Civita connection of the metric, even if $\phi$ is a solution of $(E_f,\Omega_H)$. To restrict to solutions with this property, it is usually imposed that $T_{\mu\nu}^\nu=0$, a situation which is reminiscent of the gauge freedom \cite{julia_currents_1998, dadhich_equivalence_2012}. It turns out that the condition $T_{\mu\nu}^\nu=0$ precisely fixes a respresentant in every class of kernel-related sections. In other words, every class of kernel-related solutions has exactly one respresentant such that the connection is the Levi-Civita connection of the metric \cite{gaset_new_2019}.

The kernel distribution $K$ is involutive, and we can consider the quotient $E_f\rightarrow \sfrac {E_f}{K}$. $\Omega$ is projectable to $\sfrac {E_f}{K}$ to a form $\Omega_k$, which is regular. It can be proven that the system $(\sfrac {E_f}{K},\Omega_K)$ is equivalent to the Hamiltonian system for the Einstein-Hilbert Lagrangian \cite{gaset_new_2019}.

\section{Conclusions and outlook}

We have used the kernel of a premultisymplecic form to classify its solutions. We have introduce the concept of expanded solutions, which generalizes the fact in mechanics that a solution plus a vector field of the kernel of the presympelctic form is another solution. In the case of variational premultisymplectic forms, there is an equivalence relation which classify the solutions in expanded solutions. We apply this results to sections, defining and equivalence relation of solutions. We also show how to reduce the system by the kernel distribution, and the solutions are related.

We also considered multisympelctic systems with extra conditions on the solutions. This leads to the definition of weak kernel vector fields. We apply the results to mechanics, electromagnetism and the Metric-Affine description of gravity.

This work is a contribution to the project of characterising the different solutions of a multisymplectic system. The study of the multiple solutions which do not arise from the degeneracy of the multisymplectic form remains open. It is a problem closely related to the reduction by symmetries, which this work may help to simplify.

\section*{Acknowledgments}

We acknowledge the financial support of the 
{\sl Ministerio de Ciencia e Innovaci\'on} (Spain), project D2021-125515NB-21.

\bibliographystyle{abbrv}
\bibliography{bibliografia.bib}

\end{document}